\documentclass[a4]{article}
\usepackage[T1]{fontenc}
\usepackage[latin1]{inputenc}
\usepackage{amsmath, amssymb, amsthm}
\usepackage[english]{babel}
\usepackage{indentfirst}
\usepackage{booktabs}
\usepackage{array}
\usepackage{caption,appendix}
\usepackage{geometry}
\usepackage{float}
\usepackage{cancel}
\usepackage{bbold}
\usepackage{authblk}
\usepackage{cite}
\numberwithin{equation}{section}
\newtheorem{theorem}{Theorem}[section]
\theoremstyle{definition}
\newtheorem{defn}{Definition}[section]
\newtheorem{lem}{Lemma}[section]
\begin{document}
\author[1,2,3]{Salvatore Capozziello \thanks{capozziello@unina.it}}
\author[4,2]{Maurizio Capriolo \thanks{mcapriolo@unisa.it} }
\author[4,5]{Gaetano Lambiase \thanks{glambiase@unisa.it}}
\affil[1]{\emph{Dipartimento di Fisica "E. Pancini", Universit\`a di Napoli {}``Federico II'', Compl. Univ. di
		   Monte S. Angelo, Edificio G, Via Cinthia, I-80126, Napoli, Italy, }}
\affil[2]{\emph{INFN Sezione  di Napoli, Compl. Univ. di
		   Monte S. Angelo, Edificio G, Via Cinthia, I-80126,  Napoli, Italy,}}
\affil[3]{\emph{Scuola Superiore Meridionale, Largo S. Marcellino 10,  I-80138,  Napoli, Italy,}}
\affil[4]{\emph{Dipartimento di Fisica Universit\`a di Salerno, via Giovanni Paolo II, 132, Fisciano, SA I-84084, Italy.} }
\affil[5]{\emph{INFN Sezione  di Napoli,Gruppo Collegato di Salerno, via Giovanni Paolo II, 132, Fisciano, SA I-84084, Italy.} }
\date{\today}
\title{\textbf{The energy--momentum complex in non-local gravity}}
\maketitle
\begin{abstract} 

In General Relativity, the issue of defining  the gravitational energy contained in a given spatial region is still unresolved, except for particular cases of localized objects where the asymptotic flatness holds for a given spacetime.  In principle, a theory of gravity is not self-consistent, if the whole energy content  is not uniquely defined in a specific volume.  Here we generalize the Einstein gravitational energy-momentum pseudotensor to  non-local theories of gravity where  analytic functions of the non-local integral operator $\Box^{-1}$ are taken into account. We  apply the Noether theorem to a gravitational Lagrangian, supposed invariant under  the one-parameter group of diffeomorphisms, that is, the infinitesimal rigid translations.  The invariance of  non-local gravitational action under global translations leads to a locally conserved Noether current, and thus, to the definition of a gravitational energy-momentum pseudotensor, which is  an affine object transforming like a tensor  under affine transformations.  Furthermore,  the energy-momentum complex  remains locally conserved, thanks to the non-local  contracted Bianchi identities. The continuity equations for the  gravitational pseudotensor and  the energy-momentum complex,  taking into account both  gravitational and matter components,  can be derived. Finally, the  weak field limit of pseudotensor is performed to lowest order in metric  perturbation   in view of astrophysical applications.
\end{abstract}

\section{Introduction}
Recently, non-local contributions to the gravitational action have been considered from various points of view as possible solutions  of the problem  of renormalization and regularization of  gravitational field ~\cite{Modesto1,Modesto2,Modesto3}.    In this context, a non-local gravitational energy-momentum pseudo-tensor can be proposed as a manifestation of non-locality of gravity and, therefore,  as a possible manifestation of quantum nature of gravity.

Theories of gravity can be endowed with  non-local properties in three ways~\cite{CBNLC}. Firstly through integral operators acting on functions whose value at a given  point depends on  the values of  fields at another point in spacetime~\cite{DWCI,DWCII,DW3}. Secondly, through gravitational Lagrangians involving an analytic non-polynomial function $\mathcal{F}$ of the operator $\Box$, which can be expanded in convergent series with real coefficients as
\begin{equation}
\mathcal{F}(\Box)=\sum_{h=1}^{\infty}a_{h}\Box^{h}\ ,
\end{equation}
known as Infinite Derivative Theories of Gravity (IDG)~\cite{BLP,BLMTY,BLY, BGLM,BCHKLM,BLM,BKLM}.  Thirdly, through a suitable constitutive law where, like in  electrodynamics,   temporal dispersion,   anisotropy and non-homogeneity of  medium, i.e. the spatial dispersion, are due to temporal and spatial non-locality, respectively~\cite{MOP,MNLT2, LLECM,JED,Chirco}.

At infra-red scales,  non-local models of gravity can naturally explain   late-time  acceleration without introducing  exotic material components such as  dark matter and dark energy~\cite{DWCI}.  In addition, they can potentially fix some  cosmological and astrophysical problems  plaguing the $\Lambda$CDM model~\cite{Acunzo, Dragovic1,Dragovic2,Nojiri1,Nojiri2},  black hole stability~\cite{Calcagni}, or stability and traversability of  wormhole solutions~\cite{SCNG}.

On the other hand, many authors such as Einstein, Tolman,  Landau, Lifshitz, Papapetrou, M\o ller and Weinberg have proposed  definitions for gravitational pseudotensor ~\cite{Xulu,Hasten,GoldbergCL,LLNCLVP,RosenEU,LESSNERM,PalmerGEM,FF,MWH},  to describe the energy and momentum of  gravitational field in General Relativity.  These prescriptions are based either on the introduction of a super-potential or on expanding  the Ricci tensor in  metric perturbation $h_{\mu\nu}$ or on manipulating the Einstein equations.  Although these definitions are different, it has been shown they coincide for  Kerr-Schild metric~\cite{ACV}. Many prescriptions for gravitational pseudotensor in higher-order curvature theories, in metric and Palatini approach, have been proposed ~\cite{CCT,CCL, ACCA,WKAM,vagenas,DBCJ,KOIV3,BDG}. Also for teleparallel gravity, it is possible to formulate self-consistent definition  of gravitational pseudotensor~\cite{Capozziello:2018qcp, Maluf}.

Here, we want  to propose a generalization of  Einstein gravitational pseudotensor to non-local   gravity models involving  $f(\Box^{-1})$ operators.  It will be derived from a variational principle using the Noether theorem applied to a gravitational Lagrangian invariant under global translations~\cite{Book}.  This object remains an affine tensor, i.e.  a pseudotensor, but it is a non-local quantity. Indeed,  its non-local corrections involve non-local $\Box^{-1}R$ terms, which assume, at a point $x$, a value  depending on the values assumed by the metric tensor $g_{\mu\nu}$ in all points of the integration domain.  Then, we show that the covariant conservation of the energy-momentum, associated to the gravitational and matter fields,  holds in non-local  $f(\Box^{-1} R)$~gravity, thanks to the non-local contracted Bianchi identities. Finally, we implement a lowest order expansion of the non-local pseudotensor, fundamental for astrophysical calculations such as the power carried by gravitational waves.

The paper is organized as follows. In the Sec.~\ref{A}, we firstly define the non-local integral operator $\Box^{-1}$, then we prove both that it is the inverse operator of d'Alembertian $\Box$ and a generalization of the Green second identity to the $\Box$-operator on the manifold. In addition, we perform the total variation of non-local gravitational action with respect to both the metric tensor and the coordinates.  Then we derive the field equations from a variational  principle.  Sec.~\ref{B} is devoted to the application of the Noether theorem to the non-local gravitational action for global translations.  The procedure  allows us to derive the related Noether current, i.e., the locally conserved energy-momentum pseudotensor of the gravitational field in non-local gravity.  Hence,  in Sec.~\ref{C},  we prove the non-local generalized contracted Bianchi identities and then  analyze the energy-momentum complex for gravitational and matter fields,  in particular its non-local nature and its conservation. In  Sec.~\ref{D}, we carry out the expansion to lower order in the metric perturbation $h_{\mu\nu}$ of  non-local gravitational energy-momentum pseudotensor. Finally, we discuss results and draw conclusions  in Sec.~\ref{E}.  

\section{Variational principle and field equations for non-local gravity}\label{A}
Let the spacetime $\mathcal{M}$ be a differentiable 4-manifold endowed with a Lorentzian metric $g$ and $\Omega$ be a i four-dimensional region in $\mathcal{M}$.  We can define the integral operator $\Box^{-1}$ as follows 
\begin{defn}
Let $G(x,x^{\prime})$ be the retarded Green function of the differential operator $\Box$, i.e., the solution of the partial differential equation
\begin{equation}\label{2.1}
\sqrt{-g(x)}\;\Box_{x}G(x,x^{\prime})=\delta^{4}(x-x^{\prime})\ ,.
\end{equation}
It is subject to retarded boundary condition, due to the  causality principle. It is
\begin{equation}\label{2.2}
G(x,x^{\prime})=0\quad\forall t< t^{\prime}\ ,
\end{equation}
with the d'Alembert operator defined as 
\begin{equation}\label{10}
\Box=\nabla^{\mu}\nabla_{\mu}=\frac{1}{\sqrt{-g}}\partial_{\sigma}\bigl(\sqrt{-g}g^{\sigma\lambda}\partial_{\lambda}\bigr)\ .
\end{equation}
If $p\in C^{\infty}_{o}(\mathbb{R}^{4})$ is an element of the space of  infinitely differentiable functions with compact support,  then the operator 
\begin{equation}
\Box^{-1}:C^{\infty}_{o}(\mathbb{R}^{4})\rightarrow C^{\infty}_{o}(\mathbb{R}^{4})\ ,
\end{equation}
is given by 
\begin{equation}\label{2}
(\Box^{-1}p)(x)=\int_{\Omega} d^{4}x^{\prime}\, \sqrt{-g(x^{\prime})}G(x,x^{\prime})p(x^{\prime})\ ,
\end{equation}
where $\Omega\subseteq\mathbb{R}^{4}$ and $supp(p)=\overline{\Omega}$.
\end{defn}
From now on, we shell identify the region $\Omega$ of the manifold $\mathcal{M}$ with its image $\phi(\Omega)$ through the chart $\phi:\Omega\subseteq \mathcal{M}\rightarrow\phi(\Omega)\subseteq \mathbb{R}^{4}$. It always exists because the manifold is differentiable and therefore covered by an atlas.  Likewise, we can identify the boundary of the region $\partial\Omega$ with the action of chart $\phi$ on it, i.e., $\phi(\partial\Omega)$. Therefore, let us  consider  functions, and more generally,  vector and tensor fields on the manifold, as defined on the open set of $\mathbb{R}^{4}$ by means of the graph $\phi$. Thus we have
\begin{theorem}\label{3.5}
Let $\Omega\subseteq\mathbb{R}^4$ be an open set and  $f, h\in C^{2}(\Omega)\cap C^{1}(\overline{\Omega})$ be two twice continuously differentiable functions in the open and once in its closure. If the boundary $\partial\Omega$ is a closed, regular and orientable three-dimensional hypersurface, then
\begin{equation}\label{3.6}
\int_{\Omega}d^{4}x\sqrt{-g}(f\,\Box h-h\,\Box f)=\int_{\partial\Omega}dS_{\mu}\sqrt{-g}(f\nabla^{\mu}h-h\nabla^{\mu}f)\ ,
\end{equation}
where $dS_{\mu}$ is the infinitesimal hypersurface element.
Under the further assumption that functions $f$ and $h$ vanish on boundary, i.e., $f=h=0$ on $\partial\Omega$, we get
\begin{equation}\label{3.7}
\int_{\Omega}d^{4}x\sqrt{-g}(f\,\Box h)=\int_{\Omega}d^{4}x\sqrt{-g}(h\,\Box f) \ .
\end{equation}
\end{theorem}
\begin{proof}
By means of the Leibniz rule applied to the functions $f$ and $h$, we find the differential identity 
\begin{equation}\label{3.8}
f\Box h=h\Box f+\nabla_{\mu}(f\nabla^{\mu}h-h\nabla^{\mu}f)\ .
\end{equation}
So  the integral~\eqref{3.6} can be written as  
\begin{equation}\label{3.9}
\int_{\Omega}d^{4}x\sqrt{-g}(f\,\Box h)=\int_{\Omega}d^{4}x\sqrt{-g}(h\,\Box f)+\int_{\Omega}d^{4}x\sqrt{-g}\nabla_{\mu}(f\nabla^{\mu}h-h\nabla^{\mu}f)\ ,
\end{equation}
that, thanks to the Gauss theorem, transforms the second volume integral of Eq~\eqref{3.9} into a surface integral as
\begin{equation*}
\int_{\Omega}d^{4}x\sqrt{-g}(f\,\Box h)=\int_{\Omega}d^{4}x\sqrt{-g}(h\,\Box f)+\int_{\partial\Omega}dS_{\mu}\sqrt{-g}(f\nabla^{\mu}h-h\nabla^{\mu}f)\ .
\end{equation*}
If $f$ and $h$ are zero on $\partial\Omega$,  then the integral on the boundary $\partial\Omega$ vanishes and we get  Eq~\eqref{3.7}.
\end{proof}
Then,  we show that $\Box^{-1}$ operator~\eqref{2} is the inverse operator of the d'Alembert operator $\Box$.  We can enunciate the following  proposition 
\begin{theorem}\label{2.3}
For all $p\in  C^{\infty}_{o}(\mathbb{R}^4)$, $\Box^{-1}$ is the inverse of $\Box$, i.e., 
\begin{equation}\label{2.31}
(\Box\Box^{-1})p=(\Box^{-1}\Box)p=\mathbb{1}p=p
\end{equation}
\begin{proof}
From the definition of product between two operator, we have
\begin{multline}
(\Box\Box^{-1})p(x)\equiv\Box(\Box^{-1}p)(x)=\Box_{x}\int_{\Omega} d^{4}x^{\prime}\, \sqrt{-g(x^{\prime})}G(x,x^{\prime})p(x^{\prime})\\
=\int_{\Omega} d^{4}x^{\prime}\, \sqrt{-g(x^{\prime})}\Box_{x}G(x,x^{\prime})p(x^{\prime})\\
=\frac{1}{\sqrt{-g(x)}}\int_{\Omega} d^{4}x^{\prime}\, \sqrt{-g(x^{\prime})}\delta^{4}(x-x^{\prime})p(x^{\prime})=p(x)\ ,
\end{multline}
where we used definition~\eqref{2} and the following identity involving the Dirac  $\delta$ distribution function
\begin{equation}\label{7.1}
f(x)=\int_{\Omega} d^{4}x^{\prime}\,\delta(x-x^{\prime})f(x^{\prime})\ ,
\end{equation}
non-null in $x\in\Omega$ and zero elsewhere. We have to  prove now the second identity in  Eq.~\eqref{2.31}, by means of Theorem~\eqref{3.5}. Hence we have
\begin{multline}
(\Box^{-1}\Box)p(x)\equiv\Box^{-1}(\Box p)(x)=\int_{\Omega} d^{4}x^{\prime}\, \sqrt{-g(x^{\prime})}G(x,x^{\prime})\Box_{x^{\prime}}p(x^{\prime})\\
=\int_{\Omega} d^{4}x^{\prime}\, \sqrt{-g(x^{\prime})}\Box_{x^{\prime}}G(x,x^{\prime})p(x^{\prime})\\
=\int_{\Omega} d^{4}x^{\prime}\, \sqrt{-g(x^{\prime})}\frac{\delta^{4}(x^{\prime}-x)}{\sqrt{-g(x^{\prime})}}p(x^{\prime})=p(x)\ ,
\end{multline}
\end{proof}
\end{theorem}
Let us now consider the following gravitational Lagrangian
\begin{equation}\label{1}
S_{g}=\frac{1}{2\chi}\int_{\Omega} d^{4}x\, \sqrt{-g}\bigl(R+R f(\Box^{-1}R)\bigr)\ ,
\end{equation}
where $f$ is an analytic function of $\Box^{-1}R$ and $\chi=8\pi G/c^4$ is a dimensional constant that measures the coupling between matter and geometry. The variation of  gravitational action~\eqref{1} with respect to both metric tensor and coordinates, denoted by $\tilde{\delta}$, reads as
\begin{multline}\label{3}
\tilde{\delta}S_{g}=\frac{1}{2\chi}\int_{\Omega} d^{4}x\,\Bigl[\delta(\sqrt{-g}R)+\delta(\sqrt{-g}R) f(\Box^{-1}R)\\
+\sqrt{-g}R\delta\left(f(\Box^{-1}R)\right)+\partial_{\mu}(R+Rf(\Box^{-1}R)\delta x^{\mu})\Bigr]\ ,
\end{multline}
where $\delta$ is the variation at fixed coordinates. 
Also,  we have to introduce a further theorem useful for the variation of  gravitational action~\eqref{3}, which allows us, under suitable assumptions, to move the $\Box^{-1}$ operator from a factor to another of the product in the integral. 
\begin{theorem}\label{5}
Let $f,h\in C^{\infty}(\Omega)$ be two infinitely differentiable functions on $\Omega\subseteq\mathbb{R}^{4}$, that is, $f,h:\Omega\rightarrow \mathbb{C}$. If $\Box^{-1}$ is the inverse integral operator of the d'Alembert operator $\Box$ as defined in~\eqref{2}, then 
\begin{equation}\label{5.5}
\int_{\Omega} d^{4}x\, \sqrt{-g(x)}f(x)\left(\Box^{-1}h\right)(x)=\int_{\Omega} d^{4}x\,\sqrt{-g(x)}h(x)\left(\Box^{-1}f\right)(x)\ .
\end{equation}
\end{theorem}
\begin{proof}
Let us  prove Theorem~\eqref{5} considering the identity~\eqref{7.1}. It follows
\begin{multline}\label{7}
\int_{\Omega} d^{4}x\, \sqrt{-g(x)}f(x)\left(\Box^{-1}h\right)(x)\\=
\int_{\Omega} d^{4}x\, \sqrt{-g(x)}\int_{\Omega^{\prime\prime}} d^{4}x^{\prime\prime}\, f(x^{\prime\prime})\delta(x-x^{\prime\prime})\int_{\Omega^{\prime}} d^{4}x^{\prime}\, \sqrt{-g(x^{\prime})}G(x^{\prime},x)h(x^{\prime})\\
=\int_{\Omega^{\prime}} d^{4}x^{\prime}\, \sqrt{-g(x^{\prime})}h(x^{\prime})\int_{\Omega^{\prime\prime}} d^{4}x^{\prime\prime}\,\left(\int _{\Omega}d^{4}x\, \sqrt{-g(x)}G(x^{\prime},x)\delta(x-x^{\prime\prime})\right)f(x^{\prime\prime})\\
=\int_{\Omega^{\prime}} d^{4}x^{\prime}\, \sqrt{-g(x^{\prime})}h(x^{\prime})\int_{\Omega^{\prime\prime}} d^{4}x^{\prime\prime}\,\sqrt{-g(x^{\prime\prime})}G(x^{\prime},x^{\prime\prime})f(x^{\prime\prime})\\
=\int_{\Omega^{\prime}} d^{4}x^{\prime}\,\sqrt{-g(x^{\prime})}h(x^{\prime})\left(\Box^{-1}f\right)(x^{\prime})\ .
\end{multline}
Here $\Omega$, $\Omega^{\prime}$ and $\Omega^{\prime\prime}$ are the same region covered by different charts.
\end{proof}

We establish, furthermore, a new relation that connects the variation of $\Box$ and that of $\Box^{-1}$.
\begin{theorem}
Let $\Box$ be the d'Alembert operator with its inverse operator $\Box^{-1}$ satisfying the identity 
\begin{equation}\label{7.2}
\Box\left(\Box^{-1}\right)=\Box^{-1}(\Box)=\mathbb{1}\ .
\end{equation}
For all $p\in C^{\infty}(\mathbb{R}^{4})$,  we get
\begin{equation}\label{6}
\bigl(\delta\,\Box^{-1}\bigr)p=-\Box^{-1}\delta(\Box)\Box^{-1}p\ ,
\end{equation}
where $\delta$ is the first variation of the operator part only.
\end{theorem}

\begin{proof}
Varying both sides of  identity~\eqref{7.2} and taking into account that  variation of the Identity operator $\mathbb{1}$ is zero, we have
\begin{equation}\label{6.1}
\delta (\Box\Box^{-1})=\delta{\mathbb{1}}=0\ ,
\end{equation}
and then, from Eq.~\eqref{7.2}, we get
\begin{equation}\label{6.2}
(\delta\Box)\Box^{-1}+\Box\delta(\Box^{-1})=0\ .
\end{equation}
By means of the action of  $\Box^{-1}$ operator on the left side of Eq.~\eqref{6.2}, we obtain 
\begin{equation}
\Box^{-1}(\delta\Box)\Box^{-1}+\delta(\Box^{-1})=0\ ,
\end{equation}
from which follows the relation~\eqref{6}. 
\end{proof}
Thanks to the above theorems, we are ready  to split  Eq.~\eqref{3} in three parts. The first part is the same as in  General Relativity 
\begin{equation}\label{3.1}
    \frac{1}{2\chi}\int_{\Omega} d^{4}x\, \delta(\sqrt{-g}R)=\frac{1}{2\chi}\int_{\Omega}  d^{4}x\, \sqrt{-g}\,G_{\mu\nu}\delta g^{\mu\nu}+\sqrt{-g}\,\nabla_{\sigma}\Bigl[g_{\mu\nu}\nabla^{\sigma}\delta g^{\mu\nu}-\nabla_{\lambda}\delta g^{\sigma\lambda}\Bigr]\ ,
\end{equation}
while the second one is
\begin{multline}\label{4}
\frac{1}{2\chi}\int_{\Omega}  d^{4}x\, \left[\delta(\sqrt{-g}R) f(\Box^{-1}R)\right]=\frac{1}{2\chi}\int_{\Omega}  d^{4}x\,  \biggl\{\sqrt{-g}\,f G_{\mu\nu}\delta g^{\mu\nu}\\
+\sqrt{-g}f\nabla_{\sigma}\Bigl[g_{\mu\nu}\nabla^{\sigma}\delta g^{\mu\nu}-\nabla_{\lambda}\delta g^{\sigma\lambda}\Bigr]\biggr\}\\
=\frac{1}{2\chi}\int_{\Omega}  d^{4}x\, \biggl\{\sqrt{-g}\bigl(G_{\mu\nu}+g_{\mu\nu}\Box-\nabla_{\mu}\nabla_{\nu}\bigr)f\delta g^{\mu\nu}\\
+\sqrt{-g}\nabla_{\sigma}\Bigl[\left(g_{\mu\nu}\nabla^{\sigma}\delta g^{\mu\nu}-\nabla_{\lambda}\delta g^{\sigma\lambda}\right)f-\left(g^{\lambda\sigma}g_{\mu\nu}\delta g^{\mu\nu}-\delta g^{\lambda\sigma}\right)\nabla_{\lambda}f\Bigr]\biggr\}\ ,
\end{multline}
where $G_{\mu\nu}$ is the Einstein tensor 
\begin{equation}
    G_{\mu\nu}=R_{\mu\nu}-\frac{1}{2}g_{\mu\nu}R\ .
\end{equation}
Finally, we have for the third part of Eq.~\eqref{3}, from Eqs.~\eqref{5.5} and \eqref{6}, the following form
\begin{multline}\label{8}
\frac{1}{2\chi}\int_{\Omega}  d^{4}x\,\sqrt{-g}R\delta\left(f(\Box^{-1}R)\right)=\frac{1}{2\chi}\int_{\Omega}  d^{4}x\,\sqrt{-g}Rf^{\prime}\delta\left(\Box^{-1}R\right)\\
=\frac{1}{2\chi}\int_{\Omega}  d^{4}x\, \Bigl[\sqrt{-g}Rf^{\prime}\left(\delta(\Box^{-1})R+\Box^{-1}[\delta R]\right)\Bigr]\\
=\frac{1}{2\chi}\int_{\Omega}  d^{4}x\, \Bigl[\sqrt{-g}Rf^{\prime}\Box^{-1}[\delta R]-\sqrt{-g}Rf^{\prime}\Box^{-1}\delta(\Box)\Box^{-1}R\Bigr]\ ,
\end{multline}
where $f^{\prime}=\frac{\partial f(\Box^{-1}R)}{\partial(\Box^{-1}R)}$.  The first piece of Eq.~\eqref{8} in the last line, from the identity~\eqref{5.5}, gives
\begin{multline}\label{9}
\frac{1}{2\chi}\int_{\Omega}  d^{4}x\, \sqrt{-g}Rf^{\prime}\Box^{-1}[\delta R]=\frac{1}{2\chi}\int_{\Omega}  d^{4}x\, \sqrt{-g}\,\Box^{-1}[Rf^{\prime}]\delta R\\
=\frac{1}{2\chi}\int_{\Omega}  d^{4}x\, \biggl\{\sqrt{-g}\,\Box^{-1}[Rf^{\prime}]R_{\mu\nu}\delta g^{\mu\nu}+\sqrt{-g}(g_{\mu\nu}\Box-\nabla_{\mu}\nabla_{\nu})\Box^{-1}[Rf^{\prime}]\delta g^{\mu\nu}\\
+\sqrt{-g}\nabla_{\sigma}\Bigl[\left(g_{\mu\nu}\nabla^{\sigma}\delta g^{\mu\nu}-\nabla_{\lambda}\delta g^{\sigma\lambda}\right)\Box^{-1}[Rf^{\prime}]-\left(g^{\lambda\sigma}g_{\mu\nu}\delta g^{\mu\nu}-\delta g^{\lambda\sigma}\right)\nabla_{\lambda}\Box^{-1}[Rf^{\prime}]\Bigr]\biggr\}\ .
\end{multline} 
While the second  piece of Eq.~\eqref{8} in the last line, by means of the  d'Alembert operator~\eqref{10}
and from Eq.~\eqref{5.5}, yields
\begin{multline}\label{11}
\frac{1}{2\chi}\int_{\Omega}  d^{4}x\,\bigl[ -\sqrt{-g}Rf^{\prime}\Box^{-1}\delta(\Box)\Box^{-1}R\bigr]\\
=\frac{1}{2\chi}\int_{\Omega}  d^{4}x\,\bigl[ -\sqrt{-g}\Box^{-1}[Rf^{\prime}]\delta(\Box)\Box^{-1}R\bigr]\\
=\frac{1}{2\chi}\int_{\Omega}  d^{4}x\,\Biggl[ -\sqrt{-g}\,\Box^{-1}[Rf^{\prime}]\delta\left(\frac{1}{\sqrt{-g}}\right)\partial_{\sigma}\bigl(\sqrt{-g}g^{\sigma\lambda}\partial_{\lambda}\bigr)\Box^{-1}R\\
-\sqrt{-g}\,\Box^{-1}[Rf^{\prime}]\frac{1}{\sqrt{-g}}\partial_{\sigma}\bigl(\delta\left(\sqrt{-g}g^{\sigma\lambda}\right)\partial_{\lambda}\bigr)\Box^{-1}R\Biggr]\\
=\frac{1}{2\chi}\int_{\Omega}  d^{4}x\, \Biggl\{\sqrt{-g}\biggl[-\frac{1}{2}g_{\mu\nu}R\,\Box^{-1}[Rf^{\prime}]\delta g^{\mu\nu}\biggr]+\partial_{\sigma}\left(\Box^{-1}[Rf^{\prime}]\right)\partial_{\lambda}\left(\Box^{-1}R\right)\delta\left(\sqrt{-g}g^{\sigma\lambda}\right)\\
-\partial_{\sigma}\biggl[\Box^{-1}[Rf^{\prime}]\partial_{\lambda}\left(\Box^{-1}R\right)\delta\left(\sqrt{-g}g^{\sigma\lambda}\right)\biggr]\Biggr\}\ .
\end{multline}
According to Eqs.~\eqref{3.1}, \eqref{4}, \eqref{9}, \eqref{11} and the following relation
\begin{equation}\label{12}
\delta\left(\sqrt{-g}g^{\sigma\lambda}\right)=\sqrt{-g}\left(\delta_{\mu}^{(\sigma}\delta_{\nu}^{\lambda)}-\frac{1}{2}g_{\mu\nu}g^{\sigma\lambda}\right)\ ,
\end{equation}
the variation of the gravitational action~\eqref{1} can be written as follows  
\begin{multline}\label{13}
\tilde{\delta}S_{g}=\frac{1}{2\chi}\int_{\Omega}  d^{4}x\, \sqrt{-g}\Biggl\{\biggl\{G_{\mu\nu}+\Bigl(G_{\mu\nu}+g_{\mu\nu}\Box-\nabla_{\mu}\nabla_{\nu}\Bigr)\Bigl[f+\Box^{-1}[Rf^{\prime}]\Bigr]\\
+\biggl[\delta_{\mu}^{(\sigma}\delta_{\nu}^{\lambda)}-\frac{1}{2}g_{\mu\nu}g^{\sigma\lambda}\biggr]\partial_{\sigma}\left(\Box^{-1}[Rf^{\prime}]\right)\partial_{\lambda}\left(\Box^{-1}R\right)\biggl\}\delta g^{\mu\nu}\\
+\sqrt{-g}\nabla_{\sigma}\biggl[\left(g_{\mu\nu}\nabla^{\sigma}\delta g^{\mu\nu}-\nabla_{\lambda}\delta g^{\sigma\lambda}\right)+\left(\delta g^{\lambda\sigma}-g^{\lambda\sigma}g_{\mu\nu}\delta g^{\mu\nu}\right)\nabla_{\lambda}\left(f+\Box^{-1}[Rf^{\prime}]\right)\\
+\left(g_{\mu\nu}\nabla^{\sigma}\delta g^{\mu\nu}-\nabla_{\lambda}\delta g^{\sigma\lambda}\right)\left(f+\Box^{-1}[Rf^{\prime}]\right)\\
-\left(\delta_{\mu}^{(\sigma}\delta_{\nu}^{\lambda)}-\frac{1}{2}g_{\mu\nu}g^{\sigma\lambda}\right)\nabla_{\lambda}\left(\Box^{-1}R\right)\Box^{-1}[Rf^{\prime}]\delta g^{\mu\nu}+\bigl(R+Rf\bigr)\delta x^{\sigma}\biggr]\Biggr\}\ .
\end{multline}
From the least action principle $\delta S_{g}=0$, if field variations and its derivatives vanish on boundary, the field equations in vacuum are obtained, i.e., 
\begin{equation}\label{14}
G_{\mu\nu}+\Delta G_{\mu\nu}=0\ ,
\end{equation}
with 
\begin{multline}\label{15}
\Delta G_{\mu\nu}=\Bigl(G_{\mu\nu}+g_{\mu\nu}\Box-\nabla_{\mu}\nabla_{\nu}\Bigr)\Bigl[f+\Box^{-1}[Rf^{\prime}]\Bigr]\\
+\biggl[\delta_{\mu}^{(\sigma}\delta_{\nu}^{\lambda)}-\frac{1}{2}g_{\mu\nu}g^{\sigma\lambda}\biggr]\partial_{\sigma}\left(\Box^{-1}[Rf^{\prime}]\right)\partial_{\lambda}\left(\Box^{-1}R\right)\ ,
\end{multline}
or if we define 
\begin{equation}\label{16}
\mathcal{G}[P](x)=\left(\Box^{-1}P\right)(x)\ ,
\end{equation}
Eq.~\eqref{14} can be rewritten as
\begin{equation}\label{17}
G_{\mu\nu}+\Bigl(G_{\mu\nu}+g_{\mu\nu}\Box-\nabla_{\mu}\nabla_{\nu}\Bigr)\Bigl[f+\mathcal{G}[Rf^{\prime}]\Bigr]
+\biggl[\delta_{\mu}^{(\sigma}\delta_{\nu}^{\lambda)}-\frac{1}{2}g_{\mu\nu}g^{\sigma\lambda}\biggr]\partial_{\sigma}\left(\mathcal{G}[Rf^{\prime}]\right)\partial_{\lambda}\left(\mathcal{G}[R]\right)=0\ .
\end{equation}
We can find the field equations in presence of matter using the following  action
\begin{equation}\label{18}
S_{m}=\frac{1}{2\chi}\int_{\Omega}  d^{4}x\, \sqrt{-g}\,\mathcal{L}_{m}\ ,
\end{equation}
and imposing the stationarity of  total action, i.e., 
\begin{equation}\label{19}
\delta (S_{g}+S_{m})=0\ ,
\end{equation}
with the matter energy-momentum tensor  defined as
\begin{equation}\label{20}
T_{\mu\nu}=-\frac{2}{\sqrt{-g}}\frac{\delta \bigl(\sqrt{-g}\mathcal{L}_{m}\bigr)}{\delta g^{\mu\nu}}\ .
\end{equation}
Hence,  the field equations in presence of matter are~\cite{DW3}
\begin{equation}\label{21}
G_{\mu\nu}+\Delta G_{\mu\nu}=\chi T_{\mu\nu}\ ,
\end{equation}
or 
\begin{equation}\label{22}
G_{\mu\nu}+\Bigl(G_{\mu\nu}+g_{\mu\nu}\Box-\nabla_{\mu}\nabla_{\nu}\Bigr)\Bigl[f+\mathcal{G}[Rf^{\prime}]\Bigr]
+\biggl[\delta_{\mu}^{(\sigma}\delta_{\nu}^{\lambda)}-\frac{1}{2}g_{\mu\nu}g^{\sigma\lambda}\biggr]\partial_{\sigma}\left(\mathcal{G}[Rf^{\prime}]\right)\partial_{\lambda}\left(\mathcal{G}[R]\right)=\chi T_{\mu\nu}\ .
\end{equation}
We shall use these considerations to derive the gravitational energy-momentum pseudotensor.

\section{Gravitational energy-momentum pseudotensor in non--local gravity}\label{B}
Let us now use the Noether theorem to derive the non-local gravitational energy-momentum pseudotensor. If the infinitesimal coordinate transformations 
\begin{equation}\label{22}
x^{\prime\mu}=x^{\mu}+\delta x^{\mu}\ ,
\end{equation}
leave the gravitational action~\eqref{1} unchanged, $\tilde{\delta}S_{g}=0$, and the domain of integration $\Omega$ can be chosen arbitrarily, by means of the variation~\eqref{13} and the assumption that the metric tensor $g_{\mu\nu}$ is solution of the field equations in vacuum~\eqref{17}, we find a conserved current $J^{\sigma}$, i.e., the Noether current~\cite{Book}, which reads as
\begin{equation}\label{23}
\begin{split}
2\chi J^{\sigma}=&R\delta x^{\sigma}-\left(g^{\mu\nu}g^{\lambda\sigma}-g^{\mu\lambda}g^{\sigma\nu}\right)\nabla_{\lambda}\delta g_{\mu\nu}\\
&+\left(g^{\mu\nu}g^{\lambda\sigma}-g^{\mu\lambda}g^{\sigma\nu}\right)\nabla_{\lambda}\left(f+\Box^{-1}[Rf^{\prime}]\right)\delta g_{\mu\nu}\\
&-\left(g^{\mu\nu}g^{\lambda\sigma}-g^{\mu\lambda}g^{\sigma\nu}\right)\left(f+\Box^{-1}[Rf^{\prime}]\right)\nabla_{\lambda}\delta g_{\mu\nu}\\
&-\left(\frac{1}{2}g^{\mu\nu}g^{\lambda\sigma}-g^{\mu\lambda}g^{\sigma\nu}\right)\nabla_{\lambda}\left(\Box^{-1}R\right)\Box^{-1}[Rf^{\prime}]\delta g_{\mu\nu}+Rf\delta x^{\sigma}\ ,
\end{split}
\end{equation}
that obeys the following local continuity equation
\begin{equation}\label{24}
\partial_{\sigma}\left(\sqrt{-g}J^{\sigma}\right)=0\ .
\end{equation}
Integrating the continuity equation~\eqref{24} over a three-dimensional volume $V$  at a given time $x^{0}$, from the  Gauss theorem, we obtain
\begin{equation}\label{24.5}
\frac{d}{dx^{0}}\int_{V}  d^{3}x\, \sqrt{-g}\,J^{0}=-\int_{\partial V}  dS_{i}\, \sqrt{-g}\,J^{i}\ .
\end{equation}
If the fields with their derivatives vanish on the boundary $\partial V$, the surface integral on the right of Eq.~\eqref{24.5} vanishes, i.e.,  there is no current crossing the boundary, and we can derive the conserved Noether charge in the volume $V$, associated to symmetries~\eqref{22}
\begin{equation}\label{24.8}
Q=\int_{V}  d^{3}x\, \sqrt{-g}\,J^{0}\ .
\end{equation}
So, if we consider the one-parameter group of diffeomorphisms for the global infinitesimal translations 
\begin{equation}\label{25}
x^{\prime\mu}=x^{\mu}+\epsilon^{\mu}\ ,
\end{equation}
the local variation $\delta$ of tensor metric $g_{\mu\nu}$ becomes 
\begin{equation}\label{26}
\delta g_{\mu\nu}=g_{\mu\nu}^{\prime}(x)-g_{\mu\nu}(x)=-g_{\mu\nu,\alpha}\epsilon^{\alpha}\ .
\end{equation}
Hence, the conserved Noether current, related to the translational symmetry~\eqref{25}, becomes the  energy-momentum density of the gravitational field, while, for isolated systems, where the spacetime is asymptotically flat at spatial infinity, the conserved Noether charge becomes the energy and momentum of the gravitational field. Therefore, the translation invariance of  gravitational action, from Eq.~\eqref{23}, gives 
\begin{equation}\label{26.1}
    \tau^{\sigma}_{\phantom{\sigma}\alpha}=\tau^{\sigma\,(GR)}_{\phantom{\sigma}\alpha}+\Delta\tau^{\sigma}_{\phantom{\sigma}\alpha}\ ,
\end{equation}
where $\tau^{\sigma\,(GR)}_{\phantom{\sigma}\alpha}$ is the Einstein pseudotensor
\begin{equation}\label{26.2}
    2\chi\tau^{\sigma\,(GR)}_{\phantom{\sigma}\alpha}=R\delta^{\sigma}_{\alpha}+\left(g^{\mu\nu}g^{\lambda\sigma}-g^{\mu\lambda}g^{\sigma\nu}\right)\bigl(g_{\mu\nu,\alpha\lambda}-\Gamma^{\beta}_{\phantom{\beta}\lambda\mu} g_{\beta\nu,\alpha})\ ,
\end{equation}
while the correction $\Delta\tau^{\sigma}_{\phantom{\sigma}\alpha}$, is the {\em gravitational energy-momentum pseudotensor} of   non-local part, i.e., 
\begin{equation}\label{27}
\boxed{
\begin{aligned}
2\chi\Delta\tau^{\sigma}_{\phantom{\sigma}\alpha}=& Rf\delta_{\alpha}^{\sigma}+\left(g^{\mu\nu}g^{\lambda\sigma}-g^{\mu\lambda}g^{\sigma\nu}\right)\bigl(g_{\mu\nu,\alpha\lambda}-\Gamma^{\beta}_{\phantom{\beta}\lambda\mu} g_{\beta\nu,\alpha}\bigr)\left(f+\Box^{-1}[Rf^{\prime}]\right)\\
&-\biggl\{\left(g^{\mu\nu}g^{\lambda\sigma}-g^{\mu\lambda}g^{\sigma\nu}\right)\nabla_{\lambda}\left(f+\Box^{-1}[Rf^{\prime}]\right)\\
&-\left(\frac{1}{2}g^{\mu\nu}g^{\lambda\sigma}-g^{\mu\lambda}g^{\sigma\nu}\right)\nabla_{\lambda}\left(\Box^{-1}R\right)\Box^{-1}[Rf^{\prime}]\biggl\}g_{\mu\nu,\alpha}
\end{aligned}
}\ .
\end{equation}
The pseudotensor~\eqref{27} has been obtained taking into account that the covariant derivative of  variation for the metric tensor is 
\begin{equation}\label{28}
\nabla_{\lambda}\delta g_{\mu\nu}=\partial_{\lambda}\delta g_{\mu\nu}-\Gamma^{\alpha}_{\phantom{\alpha}\lambda\mu}\delta g_{\alpha\nu}-\Gamma^{\alpha}_{\phantom{\alpha}\lambda\nu}\delta g_{\alpha\mu}\ .
\end{equation}
The symmetry of  Levi Civita connection leads to
\begin{equation}
    \left(g^{\mu\nu}g^{\lambda\sigma}-g^{\mu\lambda}g^{\sigma\nu}\right)\Gamma^{\beta}_{\phantom{\beta}\lambda\nu} =0\ ,
\end{equation}
and the local conservation of  pseudotensor can be read as
\begin{equation}\label{29}
\partial_{\alpha}\left(\sqrt{-g}\,\tau^{\sigma}_{\phantom{\sigma}\alpha}\right)=0\ ,
\end{equation}
being
\begin{equation}\label{30}
J^{\alpha}=\tau^{\sigma}_{\phantom{\sigma}\alpha}\epsilon^{\alpha}\ .
\end{equation}
In terms of Eq.~\eqref{16}, in more compact form, one gets
\begin{equation}\label{31}
\boxed{
\begin{aligned}
2\chi\Delta\tau^{\sigma}_{\phantom{\sigma}\alpha}=& Rf\delta_{\alpha}^{\sigma}+\left(g^{\mu\nu}g^{\lambda\sigma}-g^{\mu\lambda}g^{\sigma\nu}\right)\bigl(g_{\mu\nu,\alpha\lambda}-\Gamma^{\beta}_{\phantom{\beta}\lambda\mu} g_{\beta\nu,\alpha}\bigr)\left(f+\mathcal{G}[Rf^{\prime}]\right)\\
&-\biggl\{\left(g^{\mu\nu}g^{\lambda\sigma}-g^{\mu\lambda}g^{\sigma\nu}\right)\partial_{\lambda}\left(f+\mathcal{G}[Rf^{\prime}]\right)\\
&-\left(\frac{1}{2}g^{\mu\nu}g^{\lambda\sigma}-g^{\mu\lambda}g^{\sigma\nu}\right)\partial_{\lambda}\left(\mathcal{G}[R]\right)\mathcal{G}[Rf^{\prime}]\biggl\}g_{\mu\nu,\alpha}
\end{aligned}
}\ .
\end{equation}
It has to be emphasized that,from Eqs.~\eqref{26.1}, \eqref{26.2} and \eqref{27}, it is clear that the geometric object $\tau^{\sigma}_{\phantom{\sigma}\alpha}$ is a pseudotensor not a tensor. In other words, it transforms like a tensor under affine transformations but not under generic transformations.  So $\tau^{\sigma}_{\phantom{\sigma}\alpha}$ is at least an affine tensor. In an asymptotically flat spacetime the tensoriality is recovered and the integral~\eqref{24.8} returns to being a four-vector for asymptotic linear coordinates, that is,
\begin{equation}
P^{\alpha}=\int_{V}  d^{3}x\, \sqrt{-g}\,\tau^{\alpha}_{\phantom{\alpha}0}\ ,
\end{equation}
represents the energy and momentum in $V$ of the gravitational field.
Moreover the pseudotensor $\tau^{\sigma}_{\phantom{\sigma}\alpha}$ is a non-local object because it involves non-local terms, such as $\Box^{-1}R$ or $\Box^{-1}[Rf^{\prime}]$, whose value depends on the values assumed by the metric in the integration domain.

\section{The energy-momentum complex}\label{C}
The stationarity of  gravitational action, $\tilde{\delta}S_{g}=0$, with respect to the variation $\tilde{\delta}$,  from Eqs.~\eqref{13}, \eqref{14}, \eqref{15}, \eqref{26.1} and \eqref{27}, gives 
\begin{equation}\label{32}
\frac{1}{2\chi}\sqrt{-g}\,\bigl(G_{\mu\nu}+\Delta G_{\mu\nu}\bigr)\delta g^{\mu\nu}+\partial_{\sigma}\bigl(\sqrt{-g}\tau^{\sigma}_{\phantom{\sigma}\alpha}\epsilon^{\alpha}\bigr)=0\ .
\end{equation}
Hence, inserting the field equations in presence of  matter~\eqref{21} into Eq.~\eqref{32}, we get
\begin{equation}\label{33}
-\frac{1}{2}\sqrt{-g}\,T^{\mu\nu}\delta g_{\mu\nu}+\partial_{\sigma}\bigl(\sqrt{-g}\tau^{\sigma}_{\phantom{\sigma}\alpha}\epsilon^{\alpha}\bigr)=0\ .
\end{equation}
From rigid translations and coordinates  independence from  $\epsilon^{\alpha}$, it yields 
\begin{equation}\label{34}
\frac{1}{2}\sqrt{-g}\,T^{\mu\nu} g_{\mu\nu,\alpha}+\partial_{\sigma}\bigl(\sqrt{-g}\tau^{\sigma}_{\phantom{\sigma}\alpha}\bigr)=-\sqrt{-g}\nabla_{\sigma}T^{\sigma}_{\phantom{\sigma}\alpha}+\partial_{\sigma}\bigl(\sqrt{-g}T^{\sigma}_{\phantom{\sigma}\alpha}\bigr)+\partial_{\sigma}\bigl(\sqrt{-g}\tau^{\sigma}_{\phantom{\sigma}\alpha}\bigr)\ ,
\end{equation}
where  the identity 
\begin{equation}\label{35}
\sqrt{-g}\nabla_{\sigma}T^{\sigma}_{\phantom{\sigma}\alpha}=\partial_{\sigma}\bigl(\sqrt{-g}T^{\sigma}_{\phantom{\sigma}\alpha}\bigr)-\frac{1}{2}\sqrt{-g}\, g_{\mu\nu,\alpha}T^{\mu\nu}\ ,
\end{equation}
has been taken into account.
From Eq.~\eqref{34}, we obtain 
\begin{equation}\label{36}
\partial_{\sigma}\Bigl[\sqrt{-g}\bigl(T^{\sigma}_{\phantom{\sigma}\alpha}+\tau^{\sigma}_{\phantom{\sigma}\alpha}\bigr)\Bigr]=\sqrt{-g}\nabla_{\sigma}T^{\sigma}_{\phantom{\sigma}\alpha}\ .
\end{equation}
According to the previous considerations, it is possible to  prove  generalized contracted Bianchi identities for non-local gravity~\cite{KOIV1,KOIV2,KOIV3}. They guarantee the conservation of  energy--momentum  complex of gravitational and matter components.  Let us first demonstrate a  lemma useful for our purpose.
\begin{lem}\label{36.0}
Let $f\in C^{2}(\Omega)$ be a  twice continuously differentiable function on an open set $\Omega$ of $\mathbb{R}^{4}$, $\nabla$ be the covariant derivative, $\Box$ be the d'Alembert operator and $[,]$ be the commutator, we have 
\begin{equation}\label{36.1}
[\nabla_{\nu},\Box]f=-R_{\mu\nu}\nabla^{\mu}f\ .
\end{equation}
\end{lem}
\begin{proof}
From the commutator of two covariant derivatives $\nabla_{\mu}$ and $\nabla_{\nu}$, which acts on the contravariant  vector field $A^{\gamma}$, we get
\begin{equation}\label{36.2}
[\nabla_{\mu},\nabla_{\nu}]A^{\gamma}=R^{\gamma}_{\phantom{\gamma}\lambda\mu\nu}A^{\lambda}\ .
\end{equation}
If we set $A^{\gamma}=\nabla^{\gamma}f$ and $\gamma=\nu$ in Eq.~\eqref{36.2}, we obtain 
\begin{multline}\label{36.3}
[\nabla_{\mu},\Box]f=\nabla_{\mu}\nabla_{\nu}\nabla^{\nu}f-\nabla_{\nu}\nabla^{\nu}\nabla_{\mu}f\\
=\nabla^{\nu}[\nabla_{\mu},\nabla_{\nu}]f-[\nabla_{\mu},\nabla_{\nu}]\nabla^{\nu}f=\nabla^{\nu}[\nabla_{\mu},\nabla_{\nu}]f-R_{\mu\nu}\nabla^{\nu}f \ .
\end{multline}
Thus, the commutativity of covariant derivatives of a function, that is,
\begin{equation}
[\nabla_{\mu},\nabla_{\nu}]f=0\ ,
\end{equation}
inserted into Eq.~\eqref{36.3}, gives us the result~\eqref{36.1}.
\end{proof}
\begin{theorem}[Non-local generalized contracted Bianchi identities]
Let $G_{\mu\nu}$ be the Einstein tensor and $\Delta G_{\mu\nu}$ be the corrections to the field equations due to non-local terms as in  Eq.~\eqref{21}, then the covariant 4-divergence of their sum vanishes, i.e.,
\begin{equation}\label{37}
\nabla^{\mu}\bigl(G_{\mu\nu}+\Delta G_{\mu\nu}\bigr)=0\ .
\end{equation}
\end{theorem}
\begin{proof}We carry out the 4-divergence of Eq.~\eqref{15} and we have 
\begin{multline}\label{37.1}
\nabla^{\mu}\Delta G_{\mu\nu}=\bigl(\nabla^{\mu}G_{\mu\nu}+\nabla_{\nu}\Box-\Box\nabla_{\nu}\bigr)\bigl(f+\Box^{-1}[Rf^{\prime}]\bigr)+G_{\mu\nu}\nabla^{\mu}(f+\Box^{-1}[Rf^{\prime}])\\
+\frac{1}{2}\Bigl(\delta_{\nu}^{\lambda}\nabla^{\sigma}+\delta_{\nu}^{\sigma}\nabla^{\lambda}-g^{\sigma\lambda}\nabla_{\nu}\Bigr)\nabla_{\sigma}\Box^{-1}[Rf^{\prime}]\nabla_{\lambda}\Box^{-1}R\ .
\end{multline}
So, from the contracted Bianchi identities 
\begin{equation}\label{37.12}
\nabla^{\mu}G_{\mu\nu}=0\ ,
\end{equation}
and performing some calculations,  Eq.~\eqref{37.1} can be rewritten as follows
\begin{multline}\label{37.2}
\nabla^{\mu}\Delta G_{\mu\nu}=[\nabla_{\nu},\Box]\bigl(f+\Box^{-1}[Rf^{\prime}]\bigr)+G_{\mu\nu}\nabla^{\mu}(f+\Box^{-1}[Rf^{\prime}])\\
+\frac{1}{2}\Bigl(\Box\Box^{-1}[Rf^{\prime}]\nabla_{\nu}\Box^{-1}R+\nabla_{\sigma}\Box^{-1}[Rf^{\prime}]\nabla^{\sigma}\nabla_{\nu}\Box^{-1}R+\nabla^{\sigma}\nabla_{\nu}\Box^{-1}[Rf^{\prime}]\nabla_{\sigma}\Box^{-1}R\\
+\nabla_{\nu}\Box^{-1}[Rf^{\prime}]\Box\Box^{-1}R-\nabla_{\nu}\nabla^{\sigma}[Rf^{\prime}]\nabla_{\sigma}\Box^{-1}R-\nabla_{\sigma}\Box^{-1}[Rf^{\prime}]\nabla_{\nu}\nabla^{\sigma}\Box^{-1}R\Bigr)\ .
\end{multline}
Now,  the relation~\eqref{37.2} and  the lemma~\eqref{36.0} lead to  Eq.~\eqref{37}, that is, we find  
\begin{multline}\label{37.3}
\nabla^{\mu}\Delta G_{\mu\nu}=-R_{\mu\nu}\bigl(f+\Box^{-1}[Rf^{\prime}]\bigr)+G_{\mu\nu}\nabla^{\mu}(f+\Box^{-1}[Rf^{\prime}])\\
+\frac{1}{2}Rf^{\prime}\nabla_{\nu}\Box^{-1}R+\frac{1}{2}R\nabla_{\nu}\Box^{-1}[Rf^{\prime}]\\
=-R_{\mu\nu}f^{\prime}\nabla^{\mu}\Box^{-1}R-R_{\mu\nu}\nabla^{\mu}\Box^{-1}[Rf^{\prime}]+G_{\mu\nu}\nabla^{\mu}(f+\Box^{-1}[Rf^{\prime}])\\
+\frac{1}{2}g_{\mu\nu}Rf^{\prime}\nabla^{\mu}\Box^{-1}R+\frac{1}{2}g_{\mu\nu}R\nabla^{\mu}\Box^{-1}[Rf^{\prime}]\\
=-G_{\mu\nu}\nabla^{\mu}\bigl(f+\Box^{-1}[Rf^{\prime}]\bigr)+G_{\mu\nu}\nabla^{\mu}(f+\Box^{-1}[Rf^{\prime}])=0\ .
\end{multline}
\end{proof}
According to the field equation in presence of matter~\eqref{21},   Eq.~\eqref{37} leads to the standard covariant conservation of matter energy-momentum tensor, that is, 
\begin{equation}\label{39}
\nabla_{\mu} T^{\mu\nu}=0\ .
\end{equation}
It implicitly defines the trajectories  of particles,  that is,  the time-like metric geodesics on the spacetime manifold.
Finally,  Eq.~\eqref{36} gives the local conservation of energy-momentum complex $\mathcal{T}^{\sigma}_{\phantom{\sigma}\alpha}$ in non-local gravity, that is, the {\em continuity equation for energy-momentum complex in non-local gravity}
\begin{equation}\label{43}
\boxed{
\partial_{\sigma}\Bigl[\sqrt{-g}\bigl(T^{\sigma}_{\phantom{\sigma}\alpha}+\tau^{\sigma}_{\phantom{\sigma}\alpha}\bigr)\Bigr]=0
}\ .
\end{equation}
We can define 
\begin{equation}\label{44}
\mathcal{T}^{\sigma}_{\phantom{\sigma}\alpha}=T^{\sigma}_{\phantom{\sigma}\alpha}+\tau^{\sigma}_{\phantom{\sigma}\alpha}\ ,
\end{equation}
involving all gravitational and matter contributions.

\section{Weak field limit of  non-local gravitaty energy-momentum pseudotensor}\label{D}
Let us now  develop the  low energy limit perturbing the metric tensor $g_{\mu\nu}$ around the Minkowskian metric $\eta_{\mu\nu}$. It is 
\begin{equation}\label{45}
    g_{\mu\nu}=\eta_{\mu\nu}+h_{\mu\nu}\ ,
\end{equation}
and then, we can calculate the pseudotensor~\eqref{27} or \eqref{31} to lowest order in the perturbation $h_{\mu\nu}$, that is, up to second ordear in $h_{\mu\nu}$. Therefore we get, at the  order $h^{2}$,
\begin{equation}
    \left(\tau^{\sigma}_{\phantom{\sigma}\alpha}\right)^{(2)}=\left(\tau^{\sigma\,(GR)}_{\phantom{\sigma)}\alpha}\right)^{(2)}+\left(\Delta\tau^{\sigma}_{\phantom{\sigma}\alpha}\right)^{(2)}\ ,
\end{equation}
where the Einstein pseudo-tensor is
\begin{equation}
    2\chi \left(\tau^{\sigma\, (GR)}_{\phantom{\sigma}\alpha}\right)^{(2)}=R^{(2)}\delta_{\alpha}^{\sigma}+\left(g^{\mu\nu}g^{\lambda\sigma}-g^{\mu\lambda}g^{\sigma\nu}\right)^{(1)}g_{\mu\nu,\alpha\lambda}^{(1)}\ ,
\end{equation}
and, from Eq.~\eqref{31}, the non-local perturbation of  pseudotensor takes the  form
\begin{multline}\label{46}
    2\chi\left(\Delta\tau^{\sigma}_{\phantom{\sigma}\alpha}\right)^{(2)}=R^{(1)}f^{(1)}\delta_{\alpha}^{\sigma}+\left(g^{\mu\nu}g^{\lambda\sigma}-g^{\mu\lambda}g^{\sigma\nu}\right)^{(0)}\left(f^{(1)}+\mathcal{G}^{(1)}[Rf^{\prime}]\right)_{,\lambda}g_{\mu\nu,\alpha}^{(1)}\\
    -\left(g^{\mu\nu}g^{\lambda\sigma}-g^{\mu\lambda}g^{\sigma\nu}\right)^{(0)}\left(f^{(1)}+\mathcal{G}^{(1)}[Rf^{\prime}]\right)g_{\mu\nu,\alpha\lambda}^{(1)}\ .
\end{multline}
Then, we expand $f$ as
\begin{equation}\label{46.5}
    f\left(\mathcal{G}[R]\right)(x)=f(0)+f^{\prime}(0)\mathcal{G}[R](x)+\ldots\ ,
\end{equation}
and imposing the case $f(0)=0$,  the relation~\eqref{46.5} to the first order takes the form 
\begin{equation}
    f^{(1)}\left(\mathcal{G}[R]\right)(x)=f^{\prime}(0)\mathcal{G}^{(1)}[R](x)\ .
\end{equation}
Taking into account the following first order perturbations in a generic coordinate system, the Ricci scalar becomes
\begin{equation}\label{47}
    R^{(1)}=\left(h^{\beta\gamma}_{\phantom{\beta\gamma},\beta\gamma}-\Box^{(0)} h\right)\ ,
\end{equation}
where 
\begin{equation}\label{47.1}
\Box^{(0)}=\eta^{\alpha\beta}\partial_{\alpha}\partial_{\beta}\ ,
\end{equation}
and the  non-local operator  $\Box^{-1}$ at first order reads as 
\begin{equation}\label{48}
    \mathcal{G}^{(1)}[R](x)=\left(\Box^{-1}R\right)^{(1)}(x)=-h(x)+\widetilde{\mathcal{G}}\left[h^{\beta\gamma}_{\phantom{\beta\gamma},\beta\gamma}\right](x)\ ,
\end{equation}
where 
\begin{equation}\label{49}
    \widetilde{\mathcal{G}}\left[h^{\beta\gamma}_{\phantom{\beta\gamma},\beta\gamma}\right](x)=\int_{\Omega}  d^{4}x^{\prime}G(x,x^{\prime})\,h^{\beta\gamma}_{\phantom{\beta\gamma},\beta\gamma}(x^{\prime})\ .
\end{equation}
We have to  prove the identity~\eqref{48}.   Using  Eqs.~\eqref{2.1}, \eqref{2}, \eqref{47} and the theorem~\eqref{3.5}, it is 
\begin{multline}\label{50}
    \left(\Box^{-1}R\right)^{(1)}(x)=\int_{\Omega^{\prime}}  d^{4}x^{\prime}\sqrt{-g(x^{\prime})}^{(0)}G(x,x^{\prime}) R^{(1)}(x^{\prime})\\
    =\int_{\Omega}  d^{4}x^{\prime}\sqrt{-g(x^{\prime})}^{(0)}G(x,x^{\prime})\left(h^{\beta\gamma}_{\phantom{\beta\gamma},\beta\gamma}(x^{\prime})-\Box_{x^{\prime}} h (x^{\prime})\right)\\
    =-\int_{\Omega}  d^{4}x^{\prime}\Box_{x^{\prime}}G(x,x^{\prime})h (x^{\prime})+\int_{\Omega}  d^{4}x^{\prime}G(x,x^{\prime})\,h^{\beta\gamma}_{\phantom{\beta\gamma},\beta\gamma}(x^{\prime})\\
    =-\int_{\Omega}  d^{4}x^{\prime}\delta(x-x^{\prime})h (x^{\prime})+ \widetilde{\mathcal{G}}\left[h^{\beta\gamma}_{\phantom{\beta\gamma},\beta\gamma}\right](x)=-h(x)+\widetilde{\mathcal{G}}\left[h^{\beta\gamma}_{\phantom{\beta\gamma},\beta\gamma}\right](x)\ .
\end{multline}
Furthermore, we perform the first-order perturbation of $\mathcal{G}[Rf^{\prime}]$, namely
\begin{multline}\label{51}
    \mathcal{G}^{(1)}[Rf^{\prime}](x)=\int_{\Omega}  d^{4}\sqrt{-g(x^{\prime})}^{(0)}G(x,x^{\prime}) R^{(1)}(x^{\prime})f^{\prime (0)}[\mathcal{G}](x^{\prime})\\
    =f^{\prime}(0)\int_{\Omega}  d^{4}\sqrt{-g(x^{\prime})}^{(0)}G(x,x^{\prime}) R^{(1)}(x^{\prime})=f^{\prime}(0)   \mathcal{G}^{(1)}[R](x)\ .
\end{multline}
Finally substituting the Eqs.~\eqref{47}, \eqref{48} and \eqref{51} in the non-local perturbed gravitational energy--momentum pseudotensor~\eqref{46}, we derive the \emph{non-local corrections of the gravitational pseudo-tensor $\tau^{\sigma}_{\phantom{\sigma}\alpha}$ to the second order} in $h_{\mu\nu}$, that is,
\begin{equation}\label{52}
\boxed{
\begin{aligned}
    2\chi \left(\Delta\tau^{\sigma}_{\phantom{\sigma}\alpha}\right)^{(2)}=\biggl\{\Bigl(h^{\beta\gamma}_{\phantom{\beta\gamma},\beta\gamma}&-\Box h\Bigr)\Bigl(-h+\widetilde{\mathcal{G}}\left[h^{\beta\gamma}_{\phantom{\beta\gamma},\beta\gamma}\right]\Bigr)\delta^{\sigma}_{\alpha}\\
    &+2\left(\eta^{\mu\nu}\eta^{\lambda\sigma}-\eta^{\mu\lambda}\eta^{\nu\sigma}\right)\Bigl(-h+\widetilde{\mathcal{G}}\left[h^{\beta\gamma}_{\phantom{\beta\gamma},\beta\gamma}\right]\Bigr)_{,\lambda}h_{\mu\nu,\alpha}\\
   & \qquad -2\left(\eta^{\mu\nu}\eta^{\lambda\sigma}-\eta^{\mu\lambda}\eta^{\nu\sigma}\right)\Bigl(-h+\widetilde{\mathcal{G}}\left[h^{\beta\gamma}_{\phantom{\beta\gamma},\beta\gamma}\right]\Bigr)h_{\mu\nu,\alpha\lambda}\biggr\}f^{\prime}(0)
\end{aligned}
}\ .
\end{equation}
The non-local contribution in Eq.~\eqref{52} is evident and, as discussed in Refs.~\cite{CAPRIOLOM, CCCQG2021,CCN}, it can contribute to gravitational radiation representing a signature for non-local gravity.

\section{Discussion and Conclusions}\label{E}
 In this paper, we investigated how  non-locality  gravity induces correction terms $\Delta\tau^{\sigma}_{\phantom{\sigma}\alpha}$ into the Einstein gravitational  pseudotensor. Considering the Noether theorem and imposing the invariance of  gravitational action under rigid translations, we found the associated conserved Noether current and charge. They can be interpreted as the gravitational density of the energy-momentum and the energy and momentum of gravitational field present in a spatial volume enclosing localized massive objects. The density and flux density of the gravitational energy and momentum expressed in Eq.~\eqref{27} are not described by a covariant tensor, which means that, under general coordinate transformations, it does not transform like a tensor. The geometrical object~\eqref{27} is an affine tensor or pseudotensor because it transforms like a tensor only under affine transformations.  The non-tensorial character of  Eq.~\eqref{27} is closely linked to the non-localization of gravitational energy which holds also  in non-local gravity. The non-locality of the gravitational pseudotensor intervenes through  integral operators, like $\Box^{-1}$,  where its value, at a given point $x$,  takes into account the value assumed by the fields in other points $x^{\prime}$ of the spacetime. Then,  by generalizing the contracted Bianchi identities to the non-local gravity, we have obtained an equation of continuity for the energy-momentum complex that ensures its local conservation. Finally, we studied the behavior at low energies of the non-local corrections of the gravitational pseudotensor~\eqref{52}, expanding it up to the second order in $h_{\mu\nu}$.  The non-local gravitational energy-momentum pseudotensor is a crucial physical quantity because, thanks to the gravitational waves obtained and analyzed in the papers~\cite{CAPRIOLOM, CCCQG2021, CCN},  it is possible to calculate the power emitted by a radiative system and transported by the waves with all its polarizations and multipole terms.  The presence, in the gravitational radiation, of a scalar component with lower multipoles, in addition to the standard quadrupole tensor component,  can be investigated thanks to  the gravitational pseudotensor. In this perspective,  it can give  a relevant signature for the non-local gravity.  In a forthcoming paper, we will investigate possible observational constraints on these features.

 \section*{Acknowledgements}
 This paper is based upon work from COST Action CA21136 {\it Addressing
observational tensions in cosmology with systematics and fundamental physics} (CosmoVerse) supported by COST (European Cooperation in Science and Technology). 
 Authors acknowledge the Istituto Nazionale di Fisica Nucleare (INFN) Sez. di Napoli,  Iniziative Specifiche QGSKY and MOONLIGHT,   and the Istituto Nazionale di Alta Matematica (INdAM), gruppo GNFM.

\end{document}